\documentclass[aps,nopacs,twocolumn,10pt]{revtex4}

\usepackage{amsmath}
\usepackage[titletoc,title]{appendix}
\usepackage{graphicx,epic,eepic,epsfig,amsmath,latexsym,amssymb,verbatim,color}
\usepackage{dsfont}
\usepackage{MnSymbol}

 \usepackage{tikz}

 \usepackage{hyperref}
\hypersetup{colorlinks=true,citecolor=blue,linkcolor=blue,filecolor=blue,urlcolor=blue,breaklinks=true}

\usepackage[marginal]{footmisc}
\usepackage{url}
\usepackage{theorem}
\newtheorem{definition}{Definition}
\newtheorem{proposition}[definition]{Proposition}
\newtheorem{lemma}[definition]{Lemma}

\newtheorem{theorem}[definition]{Theorem}

\def\squareforqed{\hbox{\rlap{$\sqcap$}$\sqcup$}}
\def\qed{\ifmmode\squareforqed\else{\unskip\nobreak\hfil
\penalty50\hskip1em\null\nobreak\hfil\squareforqed
\parfillskip=0pt\finalhyphendemerits=0\endgraf}\fi}
\def\endenv{\ifmmode\;\else{\unskip\nobreak\hfil
\penalty50\hskip1em\null\nobreak\hfil\;
\parfillskip=0pt\finalhyphendemerits=0\endgraf}\fi}
\newenvironment{proof}{\noindent \textbf{{Proof~} }}{\qed}
\newenvironment{remark}{\noindent \textbf{{Remark~}}}{}

\mathchardef\ordinarycolon\mathcode`\:
\mathcode`\:=\string"8000
\def\vcentcolon{\mathrel{\mathop\ordinarycolon}}
\begingroup \catcode`\:=\active
  \lowercase{\endgroup
  \let :\vcentcolon
  }

\newcommand{\nc}{\newcommand}
\nc{\rnc}{\renewcommand}
\nc{\beg}{\begin{equation}}
\nc{\eeq}{{\end{equation}}}
\nc{\beqa}{\begin{eqnarray}}
\nc{\eeqa}{\end{eqnarray}}
\nc{\lbar}[1]{\overline{#1}}
\nc{\bra}[1]{\langle#1|}
\nc{\ket}[1]{|#1\rangle}
\nc{\ketbra}[2]{|#1\rangle\!\langle#2|}
\nc{\braket}[2]{\langle#1|#2\rangle}

\nc{\proj}[1]{| #1\rangle\!\langle #1 |}
\nc{\avg}[1]{\langle#1\rangle}
\nc{\Rank}{\operatorname{Rank}}
\nc{\smfrac}[2]{\mbox{$\frac{#1}{#2}$}}
\nc{\tr}{\operatorname{Tr}}
\nc{\ox}{\otimes}
\nc{\dg}{\dagger}
\nc{\dn}{\downarrow}
\nc{\cA}{{\cal A}}
\nc{\cB}{{\cal B}}
\nc{\cC}{{\cal C}}
\nc{\cD}{{\cal D}}
\nc{\cE}{{\cal E}}
\nc{\cF}{{\cal F}}
\nc{\cG}{{\cal G}}
\nc{\cH}{{\cal H}}
\nc{\cI}{{\cal I}}
\nc{\cJ}{{\cal J}}
\nc{\cK}{{\cal K}}
\nc{\cL}{{\cal L}}
\nc{\cM}{{\cal M}}
\nc{\cN}{{\cal N}}
\nc{\cO}{{\cal O}}
\nc{\cP}{{\cal P}}
\nc{\cQ}{{\cal Q}}
\nc{\cR}{{\cal R}}
\nc{\cS}{{\cal S}}
\nc{\cT}{{\cal T}}
\nc{\cV}{{\cal V}}
\nc{\cX}{{\cal X}}
\nc{\cY}{{\cal Y}}
\nc{\cZ}{{\cal Z}}
\nc{\cW}{{\cal W}}
\nc{\csupp}{{\operatorname{csupp}}}
\nc{\qsupp}{{\operatorname{qsupp}}}
\nc{\var}{{\operatorname{var}}}
\nc{\rar}{\rightarrow}
\nc{\lrar}{\longrightarrow}
\nc{\polylog}{{\operatorname{polylog}}}
\nc{\1}{{\mathds{1}}}
\nc{\wt}{{\operatorname{wt}}}
\nc{\av}[1]{{\left\langle {#1} \right\rangle}}
\nc{\supp}{{\operatorname{supp}}}

\nc{\RR}{{{\mathbb R}}}
\nc{\CC}{{{\mathbb C}}}
\nc{\FF}{{{\mathbb F}}}
\nc{\NN}{{{\mathbb N}}}
\nc{\ZZ}{{{\mathbb Z}}}
\nc{\PP}{{{\mathbb P}}}
\nc{\QQ}{{{\mathbb Q}}}
\nc{\UU}{{{\mathbb U}}}
\nc{\EE}{{{\mathbb E}}}
\nc{\id}{{\operatorname{id}}}

\nc{\CHSH}{{\operatorname{CHSH}}}

\nc{\be}{\begin{equation}}
\nc{\ee}{{\end{equation}}}
\nc{\bea}{\begin{eqnarray}}
\nc{\eea}{\end{eqnarray}}
\nc{\<}{\langle}
\rnc{\>}{\rangle}
\nc{\Hom}[2]{\mbox{Hom}(\CC^{#1},\CC^{#2})}
\nc{\rU}{\mbox{U}}

\nc{\ob}[1]{#1}

\nc{\SEP}{{\text{SEP}}}
\nc{\NS}{{\text{NS}}}
\nc{\LOCC}{{\text{LOCC}}}
\nc{\PPT}{{\text{PPT}}}
\nc{\EXT}{{\text{EXT}}}
\nc{\Sym}{{\operatorname{Sym}}}

\nc{\ERLO}{{E_{\text{r,LO}}}}
\nc{\ERLOCC}{{E_{\text{r,LOCC}}}}
\nc{\ERPPT}{{E_{\text{r,PPT}}}}
\nc{\ERLOCCinfty}{{E^{\infty}_{\text{r,LOCC}}}}
\nc{\Aram}{{\operatorname{\sf A}}}

\begin{document}
\title{Approximate broadcasting of quantum correlations}

\author{Wei Xie$^{1}$}
\email{xievvvei@gmail.com}
\author{Kun Fang$^{1}$}
\email{kun.fang-1@student.uts.edu.au}
\author{Xin Wang$^{1}$}
\email{xin.wang-8@student.uts.edu.au}
\author{Runyao Duan$^{1,2}$}
\email{runyao.duan@uts.edu.au}
\affiliation{$^1$Centre for Quantum Software and Information, Faculty of Engineering and Information Technology, University of Technology Sydney, NSW 2007, Australia}
\affiliation{$^2$UTS-AMSS Joint Research Laboratory for Quantum Computation and Quantum Information Processing, Academy of Mathematics and Systems Science, Chinese Academy of Sciences, Beijing 100190, China}

\begin{abstract}
Broadcasting quantum and classical information is a basic task in quantum information processing, and is also a useful model in the study of quantum correlations including quantum discord.  We establish a full operational characterization of two-sided quantum discord in terms of bilocal broadcasting of quantum correlations. Moreover, we show that both the optimal fidelity of unilocal broadcasting of the correlations in an arbitrary bipartite quantum state and that of broadcasting an arbitrary set of quantum states can be formulized as semidefinite programs (SDPs), which are efficiently computable. We also analyze some properties of these SDPs and evaluate the broadcasting fidelities for some cases of interest.
\end{abstract}
\date{\today}
\maketitle 

\section{Introduction}
Copying information is a rather simple task in the classical realm, but unfortunately not in the quantum realm. It is not allowed to create an identical copy of an arbitrary unknown pure quantum state due to the no-cloning theorem \cite{wootters1982single,dieks1982communication}. One can clone a set of pure states if and only if they are orthogonal. The no-broadcasting theorem \cite{barnum1996noncommuting} generalizes this result to mixed states, saying that a set of quantum states can be broadcast if and only if the states commute with each other.

These no-go theorems can be further extended to the setting of local broadcast for composite quantum systems. Given a bipartite quantum state $\rho_{AB}$ shared by Alice and Bob, their objective is to perform local operations only (without communication)  to produce a state $\widehat{\rho}_{A_1A_2B_1B_2}=(\Lambda_{A\to A_1A_2}\ox \Gamma_{B\to B_1B_2})\rho_{AB}$ such that $\tr_{A_1B_1}\widehat{\rho}_{A_1A_2B_1B_2}=\tr_{A_2B_2}\widehat{\rho}_{A_1A_2B_1B_2}=\rho_{AB}$ (see Section \ref{section:prelim} for notational convention). 
It is shown in \cite{piani2008no} that this task can only be performed if $\rho_{AB}$ is classically correlated. Even if the task is relaxed to obtain two bipartite states with the same correlation as $\rho_{AB}$ (measured by the mutual information), it is feasible to do the task if and only if the given state $\rho_{AB}$ is classically correlated. This is called the no-local-broadcasting theorem \cite{piani2008no}. Furthermore, when the local operations ar only allowed for one party (e.g., Alice), the task can be done if and only if $\rho_{AB}$ is classical on $A$ \cite{luo2010decomposition,luo2010quantum,piani2017local}.

When the task of perfect broadcasting cannot be accomplished, it is natural to ask whether the broadcasting can be performed in an approximate fashion, and how to design the optimal broadcasting operation. We shall study the approximate broadcasting of states and correlations by utilizing semidefinite programs (SDPs). In Ref. \cite{piani2016hierarchy} the Bose-symmetric channel is considered as unilocal broadcasting operation and an SDP is derived for this problem. Semidefinite programming optimization techniques \cite{boyd2004convex} have found many applications to the theory of quantum information and computation (see, e.g.,  \cite{Doherty2002,Rains2001,Wang2016,jain2011qip,wang2016irreversibility,Berta2015,Li2017,kempe2010unique,wang2016semidefinite}), and also to the study of quantum correlations (see, e.g., \cite{skrzypczyk2014quantifying,navascues2008convergent,piani2016hierarchy,Napoli2016}).

Quantum discord (see Section \ref{section:discord} for definition), as an indispensable measure of quantum correlation beyond entanglement, is introduced in \cite{ollivier2001quantum} and \cite{henderson2001classical} independently. It is argued \cite{datta2008quantum} that quantum discord is responsible for the quantum speed-up over classical algorithms. Quantum discord is a quite useful concept in many fields of quantum information processing, such as local broadcasting of correlations \cite{piani2008no,brandao2015generic}, quantum computing \cite{knill1998power}, quantum data hiding \cite{piani2014quantumness}, quantum data locking \cite{boixo2011quantum}, entanglement distribution \cite{chuan2012quantum,streltsov2012quantum}, common randomness distillation \cite{devetak2004distilling}, quantum state merging \cite{cavalcanti2011operational,madhok2011interpreting,madhok2013quantum}, entanglement distillation \cite{madhok2013quantum,streltsov2011linking}, superdense coding \cite{madhok2013quantum}, quantum teleportation \cite{madhok2013quantum}, quantum metrology \cite{girolami2014quantum}, and quantum cryptography \cite{pirandola2014quantum}. Quantum discord has become an active research topic over the past few years \cite{modi2012classical,adesso2016measures}.

The local broadcasting paradigm can provide a natural operational interpretation to quantum discord. Remarkably, the minimum average loss of mutual information resulting from local operation $\Lambda_{A\to A_1\cdots A_n}$ on $A$ for arbitrary quantum state $\rho_{AB}$ approaches the quantum discord $D_A(\rho_{AB})$ of $\rho_{AB}$ as $n$ goes to infinity. This result is established in Ref. \cite{brandao2015generic} and it generalizes the work in  Ref. \cite{streltsov2013quantum} which considers pure states $\rho_{AB}$ only.  However, it remains open whether there is an analogous connection for the two-sided setting of redistributing correlations \cite{adesso2016measures}.

In this paper, we study the approximate broadcasting of quantum correlations in both asymptotic and non-asymptotic settings. In the asymptotic regime, we rigorously prove the conjecture in Ref. \cite{adesso2016measures} and show an operational meaning of the two-sided discord in terms of bilocal broadcasting of correlations; that is, the asymptotic minimum average loss of correlation after optimal bilocal broadcasting is exactly the two-sided quantum discord of the initial state. In the non-asymptotic regime, we give an alternate derivation for the SDP characterization of the optimal unilocal broadcasting fidelity and show that the universal quantum clone machine (UQCM) can also serve as the optimal universal unilocal broadcasting operation. Moreover, the optimal state-dependent unilocal broadcasting operation for pure two-qubit states is analytically solved. Similarly, we establish the SDP for the optimal broadcasting fidelity of a finite set of quantum states.

\section{Preliminaries}\label{section:prelim}
A quantum system $A$ is associated to a Hilbert space $\cH_A$ of dimension $|A|$ with some fixed orthonormal basis $\{\ket{j}_A\}_j$. In this work, we only deal with finite-dimensional spaces, and the spaces of systems with the same letter are always assumed to be isomorphic, for example, $\cH_A\cong\cH_{\widetilde A}\cong\cH_{A_1}\cong\cH_{A_2}$. The linear operators from $\cH_A$ to $\cH_B$ are always written with subscripts identifying the systems involved, for example, $X_{A\to B}$. We denote $\cS (A)$ as the set of density operators \cite{Nielsen:2011:QCQ:1972505} on system $A$.

A quantum operation (or channel) $\cE_{A\to B}$ with input system $A$ and output system $B$ is a completely positive (CP), trace preserving (TP) linear map from the linear operators on $\cH_A$ to the linear operators on $B$. A quantum-to-classical channel $\cF$ is a cptp map such that $\cF(\cdot)=\sum_{j}\tr(M_j \cdot)\ketbra{j}{j}$, where $\{M_j\}_j$ is a POVM. The set of all quantum-to-classical channels is denoted by QC. Since the subsript of an operator or operation specifies its input and output systems, we can write a product of operators or operations without the $\otimes$ symbol, and omit the identity operator or operation $\1$, which would make no confusion, for example, $X_{AB}Y_{BC}\equiv (X_{AB}\ox \1_C)(\1_A\ox Y_{BC})$ and $\cE_{B\to C}(X_{AB})\equiv (\1_A\ox\cE_{B\to C})X_{AB}$. 

The Choi-Jamio{\l}kowski matrix \cite{jamiolkowski1972linear,choi1975completely} of a quantum operation $\cE_{A\to B}$ is $J_{\cE}=(\1_{\widetilde A\to \widetilde A}\ox\cE_{A\to B})\phi_{\widetilde A A}$, where $\phi_{\widetilde A A}=\sum_{ij}\ketbra{ii}{jj}$ is the unnormalized maximally entangled state. The output of the channel $\cE_{A\to B}$ with input $\rho_A$ can be recovered from $J_{\cE}$ by $\cE_{A\to B}(\rho_A)=\tr_A(J_{\cE}^{T_A}\rho_A)$, where $T_A$ denotes the partial transpose on $A$.

We use $H(\cdot)$ to denote the von Neumann entropy of quantum states, $H(A|B):=H(AB)-H(B)$ the conditional quantum entropy, $I(A:B):=H(A)+H(B)-H(AB)$ the quantum mutual information. The fidelity $F(\rho,\sigma)=\tr\sqrt{\sqrt{\rho}\sigma\sqrt{\rho}}$, as a measure of similarity between quantum states, can be viewed as the optimal solution to an SDP \cite{killoran2012entanglement,watrous2012simpler}. The diamond norm can be used to give the distance of two quantum operations $\cE,\cF$, that is, $\|\cE-\cF\|_\diamond=\sup\{\|((\cE-\cF)\otimes \1)X\|_1:\|X\|_1=1\}$, where $\|\cdot\|_1$ is the trace norm. In addition, we denote $[n]=\{1,\dots,n\}$, and denote by $|\cdot|$ the cardinality of a set or the dimension of a linear space.

\begin{figure}[ht] 
\hspace{-1.cm}
\begin{minipage}{.3\textwidth}
\includegraphics[scale=0.75]{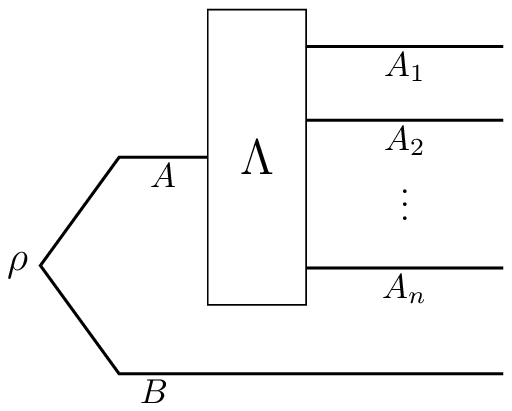}
\label{unilocal-broadcast}
\end{minipage}%
\hspace{-0.5cm}
\begin{minipage}{.25\textwidth}
\includegraphics[scale=0.75]{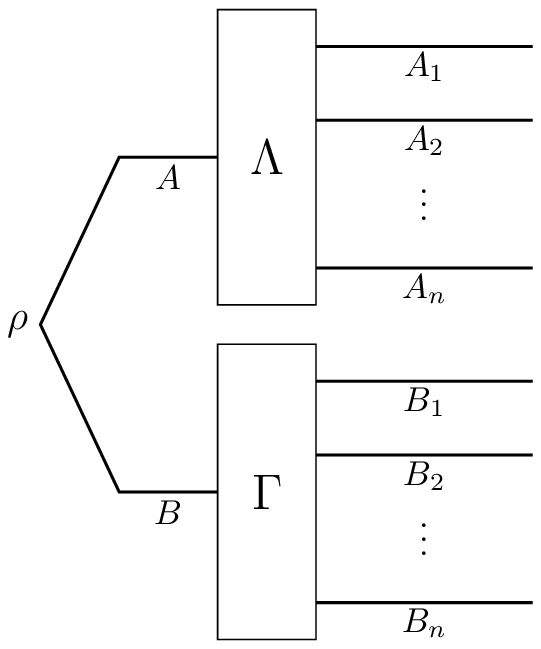}
\label{bilocal-broadcast}
\end{minipage}
\caption{Unilocal (left) and bilocal (right) broadcasting of quantum correlations in initial state $\rho^{}_{AB}$. The objective is for the quantum channels $\Lambda,\Gamma$ to make the states on $A_iB$ or $A_iB_i$ as close to $\rho_{AB}^{}$ as possible measured in some way.}
\label{fig:broadcast}
\end{figure}

\section{Asymptotic bilocal broadcasting and two-sided quantum discord}\label{section:discord}

\subsection{Previous results}
The concept of quantum discord was introduced by \cite{ollivier2001quantum,henderson2001classical}. Here the one-sided and two-sided quantum discord of a bipartite state $\rho_{AB}$ are defined by
\begin{equation}\label{def-discord1} 
D_A(\rho_{AB})  := \min_{\cE_A\in\text{QC}} (I(A:B)_{\rho_{AB}}-I(A:B)_{(\cE_A\ox\1_B) \rho_{AB}}), 
\end{equation}
\begin{equation} \label{def-discord2}
D_{AB}(\rho_{AB}) := \min_{\cE_A,\cF_B\in\text{QC}} (I(A:B)_{\rho_{AB}}-I(A:B)_{(\cE_A\ox\cF_B) \rho_{AB}}),
\end{equation}
respectively.

It is shown in \cite{brandao2015generic} that the one-sided quantum discord is equal to the asymptotic average loss of correlation after the optimal broadcasting operation. Consider the following scenario. Alice and Bob, apart away from each other, share a bipartite quantum state $\rho_{AB}$. The information, or correlation, shared by them is measured by quantum mutual information in what follows. The goal of Alice is to broadcast the mutual information between them to many, say $n$, recipients, using local operation only. If the state is not classical on $A$, she cannot perform the task perfectly \cite{luo2010decomposition,luo2010quantum,piani2017local}, and the mutual information between each recipient and Bob would decrease in general. Now her task naturally becomes to design a broadcasting operation in order to minimize the average loss of mutual information. Remarkably, the minimal average loss of correlation approaches quantum discord $D_A$ of $\rho_{AB}$ as $n$ tends to infinity, as revealed in the following Proposition.

\begin{proposition}[\cite{brandao2015generic}]
Let $\rho_{AB}$ be a bipartite state and $D_A$ is defined by Eq. (\ref{def-discord1}). Let $\Lambda_{A\to A_1\dots A_n}$ be a cptp map and $\Lambda_j:=\tr_{\backslash A_j}\circ \Lambda$. Then
\begin{equation*}\begin{split}
&D_{A}(\rho_{AB}) \\
= &\lim_{n\to\infty}\min_{\Lambda_{A\to A_1\dots A_n}} \frac{1}{n}\sum_{j=1}^n \left(I(A:B)_{\rho_{AB}}-I(A_j:B)_{(\Lambda_j\ox\1_B)\rho_{AB}}\right).
\end{split}\end{equation*}
\end{proposition}

\subsection{Operational interpretation of two-sided quantum discord}
We will give an operational interpretation of two-sided quantum discord in terms of bilocal broadcasting, analogous to the case of one-sided quantum discord (see Fig. \ref{fig:broadcast}).
\begin{theorem}
\label{result}
Let $\rho_{AB}$ be a bipartite state, and the two-sided quantum discord $D_{AB}(\rho_{AB})$ is defined by Eq. (\ref{def-discord2}). Let $\Lambda_{A\to A_1\dots A_n}$ and $\Gamma_{B\to B_1\dots B_n}$ be cptp maps, and denote $\Lambda_j:=\tr_{\backslash A_j}\circ \Lambda$ and $\Gamma_j:=\tr_{\backslash B_j}\circ \Gamma$. Then
\begin{equation*}\begin{split}
&D_{AB}(\rho_{AB}) \\
=& \lim_{n\to\infty}\min_{\substack{\Lambda_{A\to A_1\dots A_n} \\ \Gamma_{B\to B_1\dots B_n}}} \frac{1}{n}\sum_{j=1}^n \left(I(A:B)_{\rho_{AB}}-I(A_j:B_j)_{(\Lambda_j\ox\Gamma_j)\rho_{AB}}\right).
\end{split}\end{equation*}
\end{theorem}

In order to prove this theorem, we need the following result.

\begin{lemma}[\cite{brandao2015generic}]\label{brandao}
Let $\Lambda:\cS(A)\to\cS(A_1\ox\cdots\ox A_n)$ be a cptp map. Denote $\Lambda_j:=\tr_{\backslash A_j}\circ\Lambda$, and fix a number $0<\delta<1$. Then there exsits a POVM $\{E_k\}_k$ and a set $S\subset[n]$ with $|S|\ge n(1-\delta)$ such that for all $j\in S$,
\begin{equation}\label{111}
\|\Lambda_j-\cE_j\|_{\diamond} \le 3\left(\frac{\ln (2) |A|^6 \log_2 |A|}{n\delta^3}\right)^{1/3},
\end{equation}
with $\cE_j(\cdot):=\sum_k \tr(E_k\cdot) \sigma_{j,k}$ for states $\sigma_{j,k}\in\cS(A_j)$. Here $|A|$ is the dimension of the space $A$.
\end{lemma}

We also need the continuity bound of mutual information. Let $\rho_{AB},\sigma_{AB}$ be bipartite states on system A, of dimension $|A|\ge 2$, and system B. Assume $\gamma:=\frac{1}{2}\|\rho_{AB}-\sigma_{AB}\|_1\le\frac{1}{2}$. Due to Fannes-Audenaert inequality \cite{fannes1973continuity,audenaert2007sharp} and the fact that quantum operation cannot increase trace distance between two states, it holds $|H(A)_{\rho_{AB}}-H(A)_{\sigma_{AB}}| \le \frac{1}{2}\|\rho_A-\sigma_A\|_1 \log_2 (|A|-1)+h_2(\frac{1}{2}\|\rho_A-\sigma_A\|_1)\le \gamma \log_2(|A|-1)+h_2(\gamma)$, where $h_2(x):=-x\log_2 x-(1-x)\log_2(1-x)$ is the binary entropy functon. Due to Alicki-Fannes inequality \cite{alicki2004continuity} (see also \cite{winter2016tight} for a tighter continuity bound for conditional entropy), it holds $|H(A|B)_{\rho}-H(A|B)_{\sigma}|\le 8\gamma\log_2 |A| +2h_2(2\gamma)$. Therefore,
\begin{equation}\begin{split}\label{cont-bound}
&|I(A:B)_{\rho}-I(A:B)_{\sigma}|\\
\le& |H(A)_{\rho}-H(A)_{\sigma}|+|H(A|B)_{\rho}-H(A|B)_{\sigma}|\\
\le& 8\gamma\log_2 |A| + \gamma \log_2(|A|-1) +2h_2(2\gamma)+h_2(\gamma).
\end{split}\end{equation}

We are now in the position to prove Theorem \ref{result}.

\begin{proof}
The desired statement is equivalent to that 
\begin{equation*}\begin{split}
&\ \max_{\cE_A,\cF_B\in\text{QC}}I(A:B)_{(\cE_A\ox\cF_B)\rho_{AB}} \\
=&\ \lim_{n\to\infty}\max_{\Lambda,\Gamma} \frac{1}{n}\sum_{j=1}^n I(A_j:B_j)_{(\Lambda_j\ox\Gamma_j)\rho_{AB}}  .
\end{split}\end{equation*}

Assume the POVMs that achieve $I_c(A:B):=\max_{\cE_A,\cF_B\in\text{QC}}I(A:B)_{(\cE_A\ox\cF_B)\rho_{AB}}$ are $\{M_i\}_i$ on $A$ and $\{N_i\}_i$ on $B$, then one can take $\Lambda(\cdot)=\sum_i\tr(M_i\cdot)\ketbra{i}{i}^{\ox n}$ and $\Gamma(\cdot)=\sum_i\tr(N_i\cdot)\ketbra{i}{i}^{\ox n}$. It follows that $I_c(A:B)\le\max_{\Lambda,\Gamma} \frac{1}{n}\sum_{j=1}^n I(A_j:B_j)_{(\Lambda_j\ox\Gamma_j)\rho_{AB}}$, and it remains to show $I_c(A:B)\ge\lim_{n\to\infty}\max_{\Lambda,\Gamma} \frac{1}{n}\sum_{j=1}^n I(A_j:B_j)_{(\Lambda_j\ox\Gamma_j)\rho_{AB}} $.

Similarly to Lemma \ref{brandao}, let $\Gamma:\cS(B)\to \cS(B_1\ox\cdots\ox B_n)$ be an arbitrary cptp map, then there exists a POVM $\{F_k\}_k$ and a set $S'\subset [n]$ with $|S'|\ge n(1-\delta)$ such that for all $j\in S'$,
\begin{equation}\label{222}
\|\Gamma_j-\cF_j\|_{\diamond} \le 3\left(\frac{\ln (2) |B|^6 \log_2 |B|}{n\delta^3}\right)^{1/3},
\end{equation}
with $\cF_j(\cdot):=\sum_k \tr(F_k\cdot) \sigma'_{j,k}$ for states $\sigma'_{j,k}\in\cS(B_j)$. Therefore for fixed $0<\delta<1$, there exists $S''\subset [n]$ with $|S''|\ge n(1-2\delta)$ such that Eqs. (\ref{111}) and (\ref{222}) hold simultaneously for all $j\in S''$. Thus
\begin{align*}
\|\Lambda_j\ox\Gamma_j-\cE_j\ox\cF_j\|_{\diamond} &\le   \|\Lambda_j-\cE_j\|_{\diamond} + \|\Gamma_j-\cF_j\|_{\diamond}\\
&\le 6\left(\frac{\ln (2) d^6 \log_2 d}{n\delta^3}\right)^{1/3} =:\varepsilon,
\end{align*}
where $d:=\max\{|A|,|B|\}$.

For any state $\rho_{AB}$, by definition of the diamond norm, we have
\begin{equation}
\|(\Lambda_j\ox\Gamma_j)\rho_{AB}-(\cE_j\ox\cF_j)\rho_{AB}\|_1\le\|\Lambda_j\ox\Gamma_j-\cE_j\ox\cF_j\|_{\diamond}\le\varepsilon.
\end{equation}

We now have
\begin{align}
&\ I(A_j:B_j)_{(\Lambda_j\ox\Gamma_j)\rho_{AB}} \nonumber \\
\le&\ I(A_j:B_j)_{(\cE_j\ox\cF_j)\rho_{AB}} + 4\varepsilon \log_2 |A_j| + \frac{\varepsilon}{2} \log_2(|A_j|-1)\nonumber\\
&+2h_2(\varepsilon)+h_2(\varepsilon /2) \label{app-of-cont-bound} \\
\le&\ I(A_j:B_j)_{(\widetilde{\cE}_j\ox\widetilde{\cF}_j)\rho_{AB}} + 4\varepsilon \log_2 |A_j| \nonumber \\
&+ \frac{\varepsilon}{2} \log_2(|A_j|-1) +2h_2(\varepsilon)+h_2(\varepsilon/2)  \label{eq-local} \\
\le&\ I_c(A:B) + 4\varepsilon \log_2 |A_j| + \frac{\varepsilon}{2} \log_2(|A_j|-1) +2h_2(\varepsilon)+h_2(\varepsilon /2) \nonumber \\
=: &\ K, \nonumber
\end{align}
where $\widetilde{\cE}_j(\cdot):=\sum_k \tr(E_k\cdot) \ketbra{k_j}{k_j}$ and $\widetilde{\cF}_j(\cdot):=\sum_k \tr(F_k\cdot) \ketbra{k'_j}{k'_j}$, and $\{\ket{k_j}\}_k$ and $\{\ket{k'_j}\}_{k'}$ are orthonormal basis of system $A_j$ and $B_j$ respectively, Eq. (\ref{app-of-cont-bound}) follows from the continuity bound Eq. (\ref{cont-bound}), and Eq. (\ref{eq-local}) follows from the fact that local operations cannot increase mutual information.

Set $\delta=n^{-1/6}$, then as $n\to\infty$ one has $\delta,\varepsilon\to 0$ and $K\to I_c(A:B)$. It follows that
\begin{align*}
&\ \frac{1}{n}\sum_{j=1}^n I(A_j:B_j)_{(\Lambda_j\ox\Gamma_j)\rho_{AB}} \\
\le&\ \frac{1}{n}\left( (1-2\delta)n\cdot K +2\delta n\cdot 2\log_2 d' \right) \\
\to&\ K\to I_c(A:B) \text{ as } n\to\infty,
\end{align*}
where $d':=\max\{|A_j|,|B_j|\}_j$. That is,
\begin{align*}
&\ \lim_{n\to\infty}\max_{\Lambda,\Gamma} \frac{1}{n}\sum_{j=1}^n I(A_j:B_j)_{(\Lambda_j\ox\Gamma_j)\rho_{AB}} \\
\le&\ I_c(A:B),
\end{align*}
and we are done.
\end{proof}

\section{Optimal universal and state-dependent broadcasting of correlations}

We now turn to the non-asymptotic regime of the local broadcasting of quantum correlations. We first study the optimal universal unilocal broadcasting and then the optimal state-dependent unilocal broadcasting.

\subsection{Optimal universal unilocal broadcasting}
We first give a general definition for the {\em unilocal $n$-broadcasting fidelity} of a bipartite state.
\begin{definition}\label{def-broadcast-fidelity-1}
Given a bipartite state $\rho_{AB}$, the optimal unilocal $n$-broadcasting fidelity of $\rho_{AB}$ on system $A$ (see Fig. \ref{fig:broadcast}) is defined as the following optimal fidelity
\begin{equation}\begin{split}\label{eq-def-broadcast-fidelity-1}
f_n(\rho_{AB}^{})=\sup \Bigl\{ & \frac{1}{n}\sum_{j=1}^n F(\rho_{AB}^{},\tr_{\backslash A_j B}  \Lambda_{A\to A_1...A_n} (\rho_{AB}^{})): \\
& \Lambda_{A\to A_1\cdots A_n}\text{ is a quantum channel} \bigl\}.
\end{split}\end{equation}
\end{definition} 

Since the set of quantum channels is compact and the fidelity function is continuous \cite{watrous2011theory}, the supremum in Eq. (\ref{eq-def-broadcast-fidelity-1}) is attained. Define a unitary operator $W_\pi$ on systems $A_1\cdots A_n$ for each permutation $\pi\in S_n$, by the action
\[ W_\pi\ket{j_1,j_2,\dots,j_n}=\ket{j_{\pi^{-1}(1)},j_{\pi^{-1}(2)},\dots,j_{\pi^{-1}(n)}} \]
for any choice of $\ket{j_1},\ket{j_2},\dots,\ket{j_n}$. A quantum channel $\Lambda_{A\to A_1\cdots A_n}$ is called a {\em symmetric broadcasting channel}, if
\[ \Lambda(\rho)=W_\pi (\Lambda(\rho))W_\pi^\dag \]
for any $\rho\in \cS(A)$ and $\pi\in S_n$.

We notice that for any channel $\Lambda_{A\to A_1\cdots A_n}$ and $\pi\in S_n$, $\Lambda(\cdot)$ and $W_\pi(\Lambda(\cdot))W_\pi^\dag$ give the same average fidelity in Eq. (\ref{eq-def-broadcast-fidelity-1}), since
\begin{align*}
& \tr_{\backslash A_j B}  \Lambda_{A\to A_1...A_n} (\rho_{AB}^{}) \\
=& \tr_{\backslash A_{\pi^{-1}(j)} B}  W_\pi (\Lambda_{A\to A_1...A_n} (\rho_{AB}^{}))  W_\pi^\dag .
\end{align*}
Thus $\frac{1}{n!}\sum_{\pi\in S_n}  W_\pi (\Lambda(\cdot)) W_\pi^\dag$, which is a symmetric broadcasting channel, also gives the same value. So we only need to consider the supremum over symmetric broadcasting channels. In Eq. (\ref{eq-def-broadcast-fidelity-1}), when $\Lambda$  is a symmetric broadcasting channel, the summands are all the same.

Therefore, the optimal unilocal $n$-broadcasting fidelity of a bipartite state $\rho_{AB}$ on $A$ can be rewritten as
\begin{equation}\begin{split}\label{eq-def-broadcast-fidelity-2}
f_n(\rho_{AB}^{})=& \max \{  F(\rho_{AB}^{},\tr_{\backslash A_1 B}  \Lambda_{A\to A_1...A_n} (\rho_{AB}^{})): \\
& \Lambda_{A\to A_1\cdots A_n}\text{ is a symmetric broadcasting channel}\}.
\end{split}\end{equation}

It is verified that $\Lambda_{A\to A_1\cdots A_n}$ is a symmetric broadcasting channel iff its Choi matrix $J_{\Lambda}$ satisfies $J_\Lambda=W_\pi J_\Lambda W_\pi^\dag$ for any $\pi$, i.e., $J_\Lambda=\frac{1}{n!}\sum_{\pi\in S_n} W_\pi J_\Lambda W_\pi^\dag$. Using this symmetry, we give the SDP characterization for optimal unilocal broadcasting fidelity as follows.

\begin{theorem}\label{theorem2}
The optimal unilocal $n$-broadcasting fidelity of $\rho_{AB}$ on $A$ is given by the optimal solution of the following SDP,
\begin{align}\begin{split}\label{sdp-of-fn}
f_n(\rho_{AB})=\max&\  \frac{1}{2}\tr (X_{AB}+X_{AB}^\dagger) \\
\mathrm{s.t.}&\  \begin{pmatrix} \rho_{AB} & X_{AB} \\ X_{AB}^{\dag} & \tr_{\backslash A_1 B}(J^{T_A}\rho_{AB}) \end{pmatrix}\ge0, \\
&\  J_{AA_1\cdots A_n}\ge 0,\tr_{\backslash A}J_{AA_1\cdots A_n}=\1_A,\\
&\ J_{AA_1\cdots A_n}=\frac{1}{n!}\sum_{\pi\in S_n} W_\pi J_{AA_1\cdots A_n} W_\pi^\dag,
\end{split}\end{align}
where $W_\pi$ acts on $A_1\cdots A_n$.
\end{theorem}

\begin{proof}
It suffices to consider the symmetric broadcasting channels only. Let $J_{AA_1\cdots A_n}$ be the Choi matrix of $\Lambda_{A\to A_1\cdots A_n}$,  then for any $\rho_A$,
\begin{align*}
\Lambda_{A\to A_1\cdots A_n}(\rho_{A}) =\tr_A(J_{AA_1\cdots A_n}^{T_A}\rho_{A}).
\end{align*}

By linearity, for any $\rho_{AB}$,
\begin{align*}
(\Lambda_{A\to A_1\cdots A_n}\ox \1_B)\rho_{AB} =\tr_A(J_{AA_1\cdots A_n}^{T_A}\rho_{AB}),
\end{align*}
and
\begin{equation*}
\tr_{\backslash A_jB}(\Lambda\ox \1_B)\rho_{AB}=\tr_{\backslash A_jB}(J_{AA_1\cdots A_n}^{T_A}\rho_{AB}).
\end{equation*}


Now we can rewrite the optimization problem in Eq. (\ref{eq-def-broadcast-fidelity-2}) in terms of the Choi matrix of $\Lambda$ as
\begin{align}\begin{split}\label{pf-thm2}
f_n(\rho_{AB})=\max&\  F(\rho_{AB},\widehat\rho_{AB}) \\
\mathrm{s.t.}&\  \widehat\rho_{AB}=\tr_{\backslash A_1 B}(J_{AA_1\cdots A_n}^{T_A}\rho_{AB}),\\
&\  J_{AA_1\cdots A_n}\ge 0,\tr_{\backslash A}J_{AA_1\cdots A_n}=\1_A,\\
&\ J_{AA_1\cdots A_n}=\frac{1}{n!}\sum_{\pi\in S_n} W_\pi J_{AA_1\cdots A_n} W_\pi^\dag.
\end{split}\end{align}

The fidelity function $F(\rho,\sigma)$ of two states $\rho,\sigma$ is given by the optimal solution of the following SDP \cite{killoran2012entanglement,watrous2012simpler},
\begin{align}\begin{split}\label{sdp-state-fidelity}
F(\rho,\sigma)=\max&\  \frac{1}{2}\tr(X+X^\dag) \\
\mathrm{s.t.}&\  \begin{pmatrix} \rho  & X \\ X^{\dag} & \sigma \end{pmatrix}\ge 0.
\end{split}\end{align}

Combining Eqs. (\ref{pf-thm2}) and (\ref {sdp-state-fidelity}) gives the desired SDP (\ref{sdp-of-fn}).
\end{proof}

\begin{remark}
The only difference between the SDP (\ref{sdp-of-fn}) and that in Ref. \cite{piani2016hierarchy} lies in the symmetry of the broadcasting channel, that is, $J=W_\pi J W_\pi^\dag$ for any $\pi\in S_n$ is required in our SDP. In Ref. \cite{piani2016hierarchy}, it is required that $J=W_{\pi_1} J W_{\pi_2}^\dag$ for any $\pi_1,\pi_2\in S_n$ which makes sure that the output state lies in the symmetric subspace. These two SDPs are different generalization of perfect unilocal broadcasting. But the SDP (\ref{sdp-of-fn}) here has a more direct derivation, and it is clear that the optimal solution to SDP (\ref{sdp-of-fn}) is no less than that to the SDP in \cite{piani2016hierarchy}. Numerical experiments show that the two SDPs give the same optimal solution for some cases of $\rho_{AB}$, but we do not know how to give a rigorous proof or disproof for general case up to now.
\end{remark}

In the SDP (\ref{sdp-of-fn}), if $(J_{AA_1\cdots A_n},X_{AB})$ is feasible solution of $f(\rho_{AB})$, then $((\1_A\ox U^{\ox n})J_{AA_1\cdots A_n}(\1_A\ox U^{\ox n})^\dag,(U_A\ox V_B)X_{AB}(U_A\ox V_B)^\dag)$ is feasible solution of $f((U_A\ox V_B)\rho_{AB}(U_A\ox V_B)^\dag)$ for any local local unitaries $U_A$ and $V_B$. In other words, the unilocal broadcasting fidelity $f_n$ is invariant under local unitaries.

We now consider the unilocal broadcasting fidelity of a pure state $\psi_{AB}:=\ketbra{\psi}{\psi}_{AB}$, and especially the maximally entangled state, under the action of the symmetric broadcasting channel. The optimal unilocal broadcasting fidelity $f_n$ of a pure state $\psi_{AB}$ can be written as

\begin{align}\begin{split}\label{pure-local-sdp}
f_n(\psi_{AB})=\max&\  \sqrt{\tr(\widehat\rho_{AB}\psi_{AB})}  \\
\mathrm{s.t.}&\  \widehat\rho_{AB}=\tr_{\backslash A_1 B}(J_{AA_1\cdots A_n}^{T_A}\psi_{AB}),\\
&\  J_{AA_1\cdots A_n}\ge 0,\tr_{\backslash A}J_{AA_1\cdots A_n}=\1_A,\\
&\ J_{AA_1\cdots A_n}=\frac{1}{n!}\sum_{\pi\in S_n} W_\pi J_{AA_1\cdots A_n} W_\pi^\dag,
\end{split}\end{align}
where $W_\pi$ acts on $A_1\cdots A_n$.

The corresponding dual SDP is
\begin{align}\begin{split}\label{pure-local-dual}
f_n(\psi_{AB})=\min&\  \sqrt{\tr Y_A} \\
\mathrm{s.t.}&\  Y_A,Z_{AA_1\cdots A_n} \text{ Hermitian}, \\
&\ \tr_B\left(\psi_{AB}^{T_A}\psi_{A_1B}\right)-Y_A \\
&\ +Z-\frac{1}{n!}\sum_{\pi\in S_n} W_\pi^\dag Z W_\pi\le 0,
\end{split}\end{align}
where, again, $W_\pi$ acts on $A_1\cdots A_n$.

It is verified that the strong duality holds by Slater's theorem since $J_{AA_1\cdots A_n}=\1/|A|^n$ is in the relative interior of the feasible region of SDP (\ref{pure-local-sdp}). That means the optimal solutions to SDPs (\ref{pure-local-sdp}) and (\ref{pure-local-dual}) concide.

\begin{proposition}
The optimal unilocal 2-broadcasting fidelity of the maximally entangled state $\Phi_d:=\ketbra{\Phi_d}{\Phi_d}$ with $\ket{\Phi_d}=\frac{1}{\sqrt{d}}\sum_{j=0}^{d-1}\ket{jj}$ on systems $AB$ is given by
\[ f_2(\Phi_d)=\sqrt{\frac{d+1}{2d}}. \]
\end{proposition}

\begin{proof}
    We prove this proposition by explicitly constructing feasible solutions in primal and dual problem both of which can achieve the value of $\sqrt{\frac{d+1}{2d}}$.

    In the primal problem, we take
    \begin{equation}
    \label{MES optimal channel}
    J_{AA_1A_2} = \sum_{i=0}^{d-1} \ket{v_i}\bra{v_i},
     \end{equation}
    where \[\ket{v_i} = \frac{1}{\sqrt{2(d+1)}}(2\ket{i}\otimes \ket{ii} + \sum_{j\neq i} \ket{j} \otimes (\ket{ij} + \ket{ji}))\]
    This operation is also known as the universal quantum copying machine (UQCM) \cite{buvzek1996quantum,buvzek1998universal}.
    
    In the dual problem, we take \[Y_A = \frac{d+1}{2d^2}\1_d,\ Z_{AA_1A_2} = -\frac{d+1}{d^3}(d\Phi_d - I_0)\otimes \1_d,\]
    where $I_0 = \sum_{i=0}^{d-1} \ket{ii}\bra{ii}$.
\end{proof}

\begin{remark}
It is interesting that the optimal unilocal 2-broadcasting channel of the maximally entangled state is the same as the UQCM which comes from the global broadcasting setting. 
There is much progress on quantum cloning machine that has been made in the past years (see, e.g., \cite{scarani2005quantum,fan2014quantum}). For $d\otimes d$ bipartite maximally entangled state, its optimal unilocal 2-broadcasting channel is denoted as $\Upsilon_{A\rightarrow A_1A_2}^d$ with Choi matrix (\ref{MES optimal channel}) and \[\tr_{A_2}\Upsilon_{A\rightarrow A_1A_2}^d(\rho_A) = \frac{d+2}{2d+2} \rho_A + \frac{1}{2d+2}\1_d,\] is a depolarizing channel. 
     \end{remark}

Next, we will introduce a worst-case quantifier for the performance of unilocal broadcasting of a symmetric channel.

\begin{definition}
For any symmetric broadcasting channel $\Lambda_{A\to A_1\cdots A_n}$, we define the unilocal broadcasting power $\cP(\Lambda)$ of $\Lambda$ as
\begin{equation}
\cP(\Lambda):=\inf_{\rho_{AB}\in\cS(AB)}F(\rho_{AB},\tr_{\backslash A_1B}\Lambda(\rho_{AB}))
\end{equation}

\end{definition}

The unilocal broadcasting power of a symmetric broadcasting channel gives a measure of the universal unilocal broadcasting ability for symmetric broadcasting channels. The universality means it is independent of the input state. The channel with a larger value of unilocal broadcasting power is more capable of unilocal broadcasting quantum states in a universal sense.

Based on the result of optimal unilocal 2-broadcasting fidelity of maximally entangled state, we will prove that the optimal unilocal 2-broadcasting channel $\Upsilon_{A\to A_1A_2}^d$ for the maximally entangled state has the greatest power for unilocal 2-broadcasting.

\begin{lemma}\label{pure universal}
For any $d\ox d$ pure state $\ket \psi$,
\begin{equation}\label{lem-pure-univ}
f_2(\proj{\psi})\ge F(\proj{\psi},\tr_{A_2}\Upsilon_{A\to A_1A_2}^d(\proj{\psi}))\ge \sqrt{\frac{d+1}{2d}}.
\end{equation}
\end{lemma}
\begin{proof}
Consider the Schmidt decomposation $\ket\psi=\sum_i\lambda_i\ket{i}_A\ket{i}_B$, where $\{\ket i_A\}_i$ and  $\{\ket i_B\}_i$ are some orthonormal bases. Thus,

\begin{equation}\begin{split}
\rho_{out}&=\tr_{A_2}\Upsilon_{A\to A_1A_2}^d(\proj{\psi}) \\
&= \sum_{ij}\lambda_i\lambda_j\ketbra{i}{j}\ox( \frac{d+2}{2d+2}\ketbra{i}{j}+ \frac{1}{2d+2}\1_d)\\
&= \frac{d+2}{2d+2}\proj{\psi}+\sum_{ij}\frac{\lambda_i\lambda_j}{2d+2}\ketbra{i}{j}\ox\1_d.
\end{split}
\end{equation}
Then the second inequality in Eq. (\ref{lem-pure-univ}) follows from
 \begin{equation}\begin{split}
&F^2(\proj{\psi},\rho_{out})\\
=&F^2 \left(\proj{\psi}, \frac{d+2}{2d+2}\proj{\psi}+\sum_{ij}\frac{\lambda_i\lambda_j}{2d+2}\ketbra{i}{j}\ox\1_d \right)\\
=& \frac{d+2}{2d+2}+\sum_{ij}\frac{\lambda_i\lambda_j}{2d+2}\bra{\psi}(\ketbra{i}{j}\ox\1_d)\ket{\psi}\\
=& \frac{d+2}{2d+2}+\frac{\sum_id\lambda_i^4}{(2d+2)d} \\
\ge& \frac{d+2}{2d+2}+\frac{(\sum_i\lambda_i^2)^2}{(2d+2)d} = \frac{d+1}{2d}.
\end{split}
\end{equation}
\end{proof}

\begin{proposition}\label{mixed state universal}
For any $d\ox d$ mixed state $\rho$, 
$$f_2(\rho)\ge F(\rho,\tr_{A_2}\Upsilon_{A\to {A_1A_2}}^d(\rho))\ge \sqrt{\frac{d+1}{2d}}.$$
\end{proposition}
\begin{proof}
Suppose $\rho =\sum_j p_j\proj{\psi_j}$ is a pure state decomposition of $\rho$
 and $\hat\rho_j=\tr_{A_2}\Upsilon_{A\to {A_1A_2}}^d(\proj{\psi_j})$, then we have
\begin{equation}\begin{split}
 \tr_{A_2}\Upsilon_{A\to A_1A_2}^d(\rho)=&\tr_{A_2}\Upsilon_{A\to A_1A_2}^d(\sum_j p_j\proj{\psi_j})\\
 =&\sum_j p_j\tr_{A_2}\Upsilon_{A\to A_1A_2}^d(\proj{\psi_j})\\
 =&\sum_j p_j\hat\rho_j
\end{split}\end{equation}
Employing the joint concavity of fidelity, we have that
\begin{equation}\begin{split}
F(\rho, \tr_{A_2}\Upsilon_{A\to A_1A_2}^d(\rho))=&F(\sum_j p_j\proj{\psi_j},\sum_j p_j\hat\rho_j)
\\
\ge&\sum_jp_j F(\proj{\psi_j},\hat\rho_j)\\
\ge&\sum_jp_j\sqrt{\frac{d+1}{2d}}=\sqrt{\frac{d+1}{2d}},
\end{split}\end{equation}
where the last inequality uses the result in Lemma \ref{pure universal}.
\end{proof}

\begin{theorem}
$\Upsilon_{A\to {A_1A_2}}^d$ has the strongest power for unilocal 2-broadcasting in $d\ox d$ system, i.e.,
\[\max_{\Lambda} \cP(\Lambda) = \cP(\Upsilon_{A\to {A_1A_2}}^d),\]
where the maximum is taken over all symmetric broadcasting channels.
\end{theorem}
\begin{proof}
For any symmetric broadcasting channel $\Lambda_{A\to A_1A_2}$, 
we have 
\begin{equation}\begin{split}
\label{power upper bound}
\cP(\Lambda)
=&\ \inf_{\rho_{AB}\in\cS(AB)}F(\rho_{AB},\tr_{\backslash A_1B}\Lambda(\rho_{AB}))\\
\leq &\ F(\Phi_d,\tr_{\backslash A_1B}\Lambda(\Phi_d))\\
\leq &\ F(\Phi_d,\tr_{\backslash A_1B}\Upsilon_{A\to {A_1A_2}}^d(\Phi_d))\\
= &\ \sqrt{\frac{d+1}{2d}},
\end{split}
\end{equation}
where $\Phi_d$ is the maximally entangled state. The second inequality holds since $\Upsilon_{A\to {A_1A_2}}^d$ is the optimal unilocal 2-broadcasting channel for $\Phi_d$.

For the  unilocal 2-broadcasting operation $\Upsilon_{A\to {A_1A_2}}^d$, from Proposition \ref{mixed state universal}, 
we have that 
\begin{equation}\label{F min upsilon}
\cP(\Upsilon_{A\to {A_1A_2}}^d)=\sqrt{\frac{d+1}{2d}}.
\end{equation}

Combining Eqs. (\ref{power upper bound}) and (\ref{F min upsilon}), it is clear that $\Upsilon_{A\to {A_1A_2}}^d$ maximizes the unilocal broadcasting power $\cP$. Thus, it  is optimal under the setting of universal unilocal 2-broadcasting.

\end{proof}

\subsection{Optimal unilocal broadcasting for two-qubit pure state}

In the following theorem, we give analytical solution of optimal unilocal 2-broadcasting fidelity for two-qubit pure state. Since $f_n$ is invariant under local unitary, we only need to consider two-qubit pure state in the form of $\ket{\psi_\theta} = \cos\theta\ket{00} + \sin\theta\ket{11}$, $\theta \in (0,\pi/4]$ without lose of generality. 
\begin{theorem}
    For two-qubit pure state $\psi_\theta=\proj{\psi_\theta}$ with $\ket{\psi_\theta} = \cos\theta\ket{00} + \sin\theta\ket{11}$, $\theta \in (0,\pi/4]$, its optimal unilocal 2-broadcasting fidelity is given by
\begin{eqnarray*}f_2(\psi_\theta)=
\begin{cases}
 \cos^2 \theta + (\sin^2 \theta)/\sqrt{2}, &\theta \in (0,\arctan (2^{-1/4})]\cr
 (\frac{3}{2}(\cos^4 \theta + \sin^4 \theta))^{1/2}, &\theta \in (\arctan (2^{-1/4}), \pi/4] \end{cases}
\end{eqnarray*}
\end{theorem}

\begin{proof}
We prove this theorem by explicitly constructing a feasible solution in both primal and dual problem which achieves $f_2(\psi_\theta)$.

\textbf{Case 1:}
   If $\theta \in (0,\arctan (2^{-1/4})]$,
   in the primal problem, we construct feasible solution 
\begin{equation}
\label{optimal universal}
J_{AA_1A_2} = \ket{v}\bra{v},
\end{equation}
where 
\[\ket{v} = \ket{000} + \frac{1}{\sqrt{2}}\ket{101}+\frac{1}{\sqrt{2}}\ket{110}.\]
 
In the dual problem, we construct feasible solution
 \[Y_A = p \begin{pmatrix} \sqrt{2}\cos^2 \theta & 0\\
 0 & \sin^2 \theta \end{pmatrix},\ \text{where}\ p = \frac{\sqrt{2}\cos^2 \theta + \sin^2 \theta}{2}.\]
 \[Z_{AA_1A_2} = x(\ket{000}\bra{110}+ \ket{110}\bra{000}+\ket{001}\bra{111}+\ket{111}\bra{001}),\]\ \text{where} $x = {\sqrt{2}}{p\cdot \sin^2 \theta }$.
 It is easy to check that $J_{AA_1A_2}$ and $\{Y_A, Z_{AA_1A_2}\}$ are feasible solutions to SDP (\ref{pure-local-sdp}) and (\ref{pure-local-dual}).

\textbf{Case 2:}
     If $\theta \in (\arctan (2^{-1/4}), \frac{\pi}{4})$,
     in the primal problem, we construct a feasible solution
     \[J_{AA_1A_2} = \ket{v_1}\bra{v_1} + \ket{v_2}\bra{v_2},\]
     where 
    \[\ket{v_1} = \sqrt{\frac{2\tan^4 \theta -1}{6}} (\ket{001}+\ket{010}) + \sqrt{\frac{4-2\cot^4 \theta}{3}}\ket{111},\]     
    \[\ket{v_2} =   \sqrt{\frac{2\cot^4 \theta-1}{6}}(\ket{101}+\ket{110}) +\sqrt{\frac{4-2\tan^4 \theta}{3}}\ket{000}.\]

     In the dual problem, let us choose
     \[Y_A = \frac{3}{2}\begin{pmatrix}\cos^4 \theta & 0\\ 0& \sin^4 \theta \end{pmatrix},\]
     \[Z_{AA_1A_2} = x(\ket{000}\bra{110}+ \ket{110}\bra{000}+\ket{001}\bra{111}+\ket{111}\bra{001}),\] \text{where} $x = -\frac{3}{2} \sin^2 \theta \cos^2 \theta$.
 It is also easy to check that $J_{AA_1A_2}$ and $\{Y_A, Z_{AA_1A_2}\}$ are feasible solutions to SDP (\ref{pure-local-sdp}) and (\ref{pure-local-dual}).
\end{proof}

From the above proof, we can see that the optimal unilocal 2-broadcasting channel is independent of parameter $\theta$ in the first piece, that is, $\theta \in (0,\arctan (2^{-1/4})]$. We denote this channel as $\Xi$ with Choi matrix $J_{AA_1A_2}$ (\ref{optimal universal}).

We show the difference between fidelity of unilocal 2-broadcasting via channel $\Upsilon$ and $\Xi$, denoted as $f_{2,\Upsilon}(\psi_\theta)$, $f_{2,\Xi}(\psi_\theta)$ respectively, and the optimal unilocal 2-broadcasting fidelity $f_2(\psi_\theta)$ in the following Fig. \ref{fig 1}.

\begin{figure}[htbp]
\centering{\includegraphics[width = 8cm]{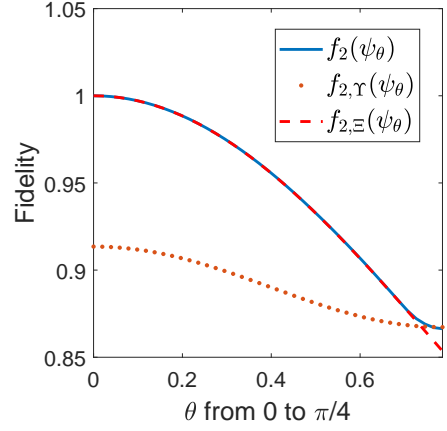}}
  \caption{The solid line depicts the optimal unilocal 2-broadcasting fidelity $f_2(\psi_\theta)$, the dashed line depicts the fidelity of unilocal 2-broadcasting via channel $\Xi$, $f_{2,\Xi}(\psi_\theta)$, which almost coincides with $f_2(\psi_\theta)$ except when $\theta$ is close to $\pi/4$, and the dotted line depicts the fidelity of unilocal 2-broadcasting via channel $\Upsilon$, $f_{2,\Upsilon}(\psi_\theta)$.}
\label{fig 1}
\vspace{-0.5cm}
\end{figure}

\section{Approximate broadcasting of a set of quantum states}
\subsection{Fidelity of  broadcasting a set of quantum states}
The no-go theorem for simultaneously broadcasting quantum states \cite{barnum1996noncommuting} says that we cannot perfectly broadcast two arbitrary noncommuting states. It is natural to ask how well we can do the task approximately. Generally, given $m$ states $\rho_i$ with respective prior probability $p_i$, how large average fidelity can we achieve when broadcasting these states via the same channel? Mathematically, assuming the given states $\rho_i$ are on the system $A$, we study how to optimize the {\em $n$-broadcasting fidelity $g_n(\eta)$ of an ensemble $\eta:=\{p_i,\rho_i\}_{i=1}^m$}, which is defined as

\begin{equation}\begin{split}
g_n(\eta):=\sup&\  \sum_{i=1}^m  \sum_{j=1}^n \frac{1}{n} p_i F(\rho_i,\widehat{\rho}_{ij})  \\
\mathrm{ s.t. }&\ \widehat\rho_{ij}=\tr_{\backslash A_j}\Lambda_{A\to A_1\cdots A_n}(\rho_i),  \\
&\ \Lambda \text{ is a quantum channel}.
\end{split}\end{equation}

Using the idea in the derivation of Eq. (\ref{eq-def-broadcast-fidelity-2}), namely, exploiting the symmetry in the broadcasting channel $\Lambda$, we can simplify this definition. The $n$-broadcasting fidelity $g_n$ of an ensemble $\eta:=\{p_i,\rho_i\}_{i=1}^m$ can be rewritten as
\begin{equation}\begin{split}
g_n(\eta)=\sup&\  \sum_{i=1}^m p_iF(\rho_i,\widehat\rho_{i1})\\
\mathrm{ s.t. }&\ \widehat\rho_{i1}=\tr_{\backslash A_1}\Lambda_{A\to A_1\cdots A_n}(\rho_i),  \\
&\ \Lambda_{} \text{ is a symmetric broadcasting channel}.
\end{split}\end{equation}

\begin{figure}[ht] 
\hspace{-1.cm}
\begin{minipage}{.3\textwidth}
\includegraphics[scale=0.85]{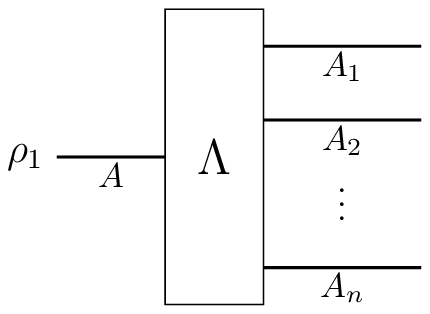}
\label{broadcast-rho}
\end{minipage}%
\hspace{-.9cm}
\begin{minipage}{.25\textwidth}
\includegraphics[scale=0.85]{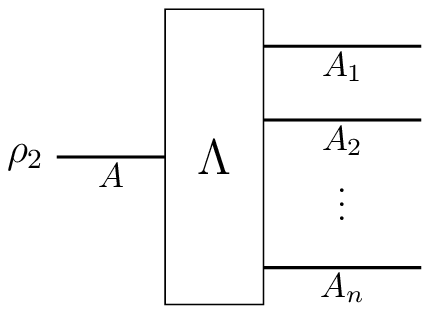}
\label{broadcast-sigma}
\end{minipage}
\caption{Broadcasting states $\rho_1,\rho_2$ via the same channel $\Lambda$.}
\label{fig:broadcast-two-states}
\end{figure}

\begin{theorem}
The $n$-broadcasting fidelity $g_n(\eta)$ of an ensemble $\eta=\{p_i,\rho_i\}_{i=1}^m$ is given by the optimal solution of the SDP in (\ref{broad-two-states-sdp}).

\end{theorem}

\begin{proof}
The output state on system $A_1$ of broadcasting $\rho_i$ is
\begin{equation}
\widehat\rho_{i1}=\tr_{\backslash A_1}\Lambda_{A\to A_1\cdots A_n}(\rho_i)=\tr_{\backslash A_1}(J_{AA_1\cdots A_n}\rho_i^T),
\end{equation}
where $J_{AA_1\cdots A_n}$ is the Choi matrix of $\Lambda_{A\to A_1\cdots A_n}$.

By using the SDP characterization of fidelity function, we then have
\begin{equation}\begin{split}\label{broad-two-states-sdp}
g_n(\eta)=\max&\ \sum_{i=1}^m \frac{1}{2}p_i \tr(X_i+X_i^\dag)\\
\mathrm{ s.t. }&\ \begin{pmatrix} \rho_i  & X_i \\ X_i^{\dag} & \tr_{\backslash A_1}(J_{AA_1\cdots A_n}\rho_i^T) \end{pmatrix}\ge 0, \forall i\in[m],\\
&\  J_{AA_1\cdots A_n}\ge 0,\tr_{\backslash A}J_{AA_1\cdots A_n}=\1_A,\\
&\ J_{AA_1\cdots A_n}=\frac{1}{n!}\sum_{\pi\in S_n} W_\pi J_{AA_1\cdots A_n} W_\pi^\dag,
\end{split}\end{equation}
where $W_\pi$ acts on $A_1\cdots A_n$.
\end{proof}

\section{Conclusions and discussion}
In summary, we have studied the approximate broadcasting of quantum correlations from several aspects. Firstly, we extend the operational characterization of one-sided quantum discord to two-sided one, that is, the asymptotic optimal average mutual information loss after the action of two local broadcasting channels is equal to the two-sided quantum discord. Then we give an alternate derivation for the SDP characterization of the unilocal broadcasting fidelity, based on which we analyze some properties of unilocal broadcasting. We show that the universal quantum clone machine (UQCM) is also the optimal universal unilocal broadcasting operation. Moreover, the optimal state-dependent unilocal broadcasting operation for pure two-qubit states is analytically solved. Finally, we also formulate the broadcasting of a finite set of quantum states as an SDP. 
It would be of interest to study other topics related to broadcasting and correlations, such as the broadcasting of Gaussian state and correlation, and the relation between Gaussian quantum broadcasting and Gaussian quantum discord. One can also study the asymptotic behavior of the $n$-broadcasting fidelity of a finite set of quantum states in the large $n$ limit.

\section*{ACKNOWLEDGEMENTS}
The authors are grateful to Kun Wang for useful discussion. This work was partly supported by the Australian Research Council under Grant Nos. DP120103776 and FT120100449.


\bibliographystyle{apsrev4-1}
\bibliography{bib}

\end{document}